\documentclass{article}

\usepackage[a4paper]{geometry}
\usepackage{amsmath}
\usepackage{amssymb}
\usepackage{amsthm}
\usepackage{multirow}
\usepackage{enumerate}
\usepackage{longtable}
\usepackage{caption2}

\allowdisplaybreaks[1]
\newcommand{\ra}[1]{\renewcommand{\arraystretch}{#1}}

\newtheorem{theorem}{Theorem}

\newtheorem{lemma}[theorem]{Lemma}
\newtheorem{proposition}[theorem]{Proposition}
\newtheorem{definition}[theorem]{Definition}

\theoremstyle{definition}
\newtheorem{example}[theorem]{Example}
\newtheorem{remark}[theorem]{Remark}

\numberwithin{theorem}{section}
\numberwithin{equation}{section}
\numberwithin{table}{section}

\newcommand{\cC}{\mathcal{C}}

\newcommand{\cL}{\mathcal{L}}
\newcommand{\cP}{\mathcal{P}}
\newcommand{\cQ}{\mathcal{Q}}

\newcommand{\de}{\delta}

\newcommand{\la}{\lambda}

\newcommand{\ol}{\overline}

\newcommand{\lf}{\lfloor}
\newcommand{\rf}{\rfloor}

\newcommand{\Rad}{\text{Rad}\,}
\newcommand{\rank}{\text{rank}}

\newcommand{\wt}{\text{wt}}

\newcommand{\Tr}{\text{Tr}}

\newcommand{\F}{\mathbb{F}}

\newcommand{\Fq}{\mathbb{F}_q}

\newcommand{\Ft}{\F_2}

\newcommand{\RM}{\text{RM}}
\newcommand{\HRM}{\text{HRM}}
\newcommand{\PRM}{\text{PRM}}

\long\def\symbolfootnote[#1]#2{\begingroup%
\def\thefootnote{\fnsymbol{footnote}}\footnote[#1]{#2}\endgroup}

\begin{document}

\title{On the weight distribution of second order Reed-Muller codes and their relatives}
\author{Shuxing Li}
\date{}

\maketitle

\symbolfootnote[0]{
S.~Li is with Faculty of Mathematics, Otto von Guericke University Magdeburg, 39106 Magdeburg, Germany (e-mail: shuxing.li@ovgu.de).
}

\begin{abstract}

The weight distribution of second order $q$-ary Reed-Muller codes have been determined by Sloane and Berlekamp (IEEE Trans. Inform. Theory, vol. IT-16, 1970) for $q=2$ and by McEliece (JPL Space Programs Summary, vol. 3, 1969) for general prime power $q$. Unfortunately, there were some mistakes in the computation of the latter one. This paper aims to provide a precise account for the weight distribution of second order $q$-ary Reed-Muller codes. In addition, the weight distributions of second order $q$-ary homogeneous Reed-Muller codes and second order $q$-ary projective Reed-Muller codes are also determined.


\end{abstract}

\section{Introduction}

Let $q$ be a prime power. Due to their elegant algebraic properties and connections to finite geometry, $q$-ary Reed-Muller codes are long-standing research objects in coding theory, see \cite[Chapters 13,14,15]{MS} for $q=2$ and \cite[Section 5]{AK} for general $q$. Moreover, second order $q$-ary Reed-Muller codes are of particular interest, since they contain some famous subcodes such as Kerdock codes \cite[Chapter 15, Section 5]{MS} and the well-rounded theory of quadratic forms over finite fields can be applied.

The weight distribution is a fundamental parameter of Reed-Muller codes. For second order $q$-ary Reed-Muller codes, their weight distributions have been computed in \cite{SB} for $q=2$ and in \cite{McE} for general $q$. Unfortunately, as observed in \cite[p. 2559]{Li}, there are some errors and typos in the computation of \cite{McE} (some essential errors are spotted in Tables 4,8,10 and typos in Tables 3,6). Hence, in this paper, we aim to provide a precise account for the weight distribution of second order $q$-ary Reed-Muller codes, with $q$ being a prime power. An outline is as follows, where the second order $q$-ary Reed-Muller is denoted by $\RM_q(2,m)$.
\begin{itemize}
\item[(1)] Observe that $\RM_q(2,m)$ is a disjoint union of cosets of the repetition code (when $q=2$) or cosets of the first order Reed-Muller code $\RM_q(1,m)$ (when $q>2$), where the coset representatives are exactly all quadratic forms from $\Fq^m$ to $\Fq$ (see \eqref{eqn-RMdecom}).
\item[(2)] When $q=2$, the weight distribution of each coset follows from the number of zeroes to a quadratic form (Proposition~\ref{prop-quadzero}). When $q>2$, the weight distribution of each coset can be derived from the results in \cite{Li} (Propositions~\ref{prop-diseven},~\ref{prop-disodd},~\ref{prop-cosetzero} and \ref{prop-cosetweight}).
     In both cases, the weight distribution of each coset depends only on the rank and type of the quadratic form, which is the coset representative. We remark that the canonical quadratic forms and related terminologies used in \cite{Li} are different from those in this paper. Thus, in order to employ the results in \cite{Li}, we need to build the correspondence between different canonical quadratic forms and terminologies at first (Table~\ref{tab-corres}).
\item[(3)] The number of quadratic forms from $\Fq^m$ to $\Fq$, with given rank and type, has been obtained by McEliece \cite{McE}, following which we can compute the frequency of each weight in $\RM_q(2,m)$.
\end{itemize}

In addition, using a similar idea, the weight distributions of second order $q$-ary homogeneous Reed-Muller codes and second order $q$-ary projective Reed-Muller codes are also computed.

Below, we recall some basic knowledge about quadratic forms over finite fields in Section~\ref{sec2}. In Section~\ref{sec3}, we compute the weight distributions of second order Reed-Muller codes, second order homogeneous Reed-Muller codes and second order projective Reed-Muller codes. Section~\ref{sec4} concludes the paper.

\section{Quadratic forms over finite fields}\label{sec2}

The mathematical mechanism behind second order Reed-Muller codes is the theory of quadratic forms over finite fields. In this section, we introduce some background knowledge about quadratic forms over finite fields.

Let $V$ be an $m$-dimensional vector space over $\Fq$. A quadratic form $Q$ defined on $V$ is a function from $V$ to $\Fq$, such that
\begin{itemize}
\item[(1)] For each ${\bf x} \in V$ and $\la \in \Fq$, $Q(\la {\bf x})=\la^2Q({\bf x})$.
\item[(2)] For each ${\bf x},{\bf y} \in V$, $Q({\bf x}+{\bf y})=Q({\bf x})+Q({\bf y})+B_Q({\bf x},{\bf y})$, where $B_Q$ is a symmetric bilinear form associated with $Q$.
\end{itemize}

For a symmetric bilinear form $B$ defined on $V$, the \emph{radical} of $B$ is defined to be
$$
\Rad B:=\{{\bf y} \in V \mid B({\bf x},{\bf y})=0, \forall {\bf x} \in V \},
$$
which is a vector space over $\Fq$. For a quadratic form $Q$ defined on $V$, the \emph{radical} of $Q$ is a vector space over $\Fq$, namely,
$$
\Rad Q:=Q^{-1}(0) \cap \Rad B_Q,
$$
where $B_Q$ is the symmetric bilinear form associated with $Q$. The \emph{rank} of $Q$ is defined to be
$$
\rank(Q)=m-\dim \Rad Q.
$$

Let $Q$ be a quadratic form defined on $V$, then $Q$ has the following unique representation
$$
Q=\begin{cases}
  \sum_{1 \le i \le j \le m} c_{ij}x_ix_j, c_{ij} \in \Fq & \mbox{if $q$ is even,} \\
  \sum_{1 \le i,j \le m} c_{ij}x_ix_j, c_{ij}=c_{ji} \in \Fq & \mbox{if $q$ is odd.}
\end{cases}
$$
Let $Q_1$ and $Q_2$ be two quadratic forms defined on $V$. They are \emph{equivalent} if there exists an $m \times m$ invertible matrix $A$ over $\Fq$, such that for each ${\bf x} \in V$,
$$
Q_1({\bf x})=Q_2(A{\bf x}).
$$
Thus, two quadratic forms are equivalent, if one can be transformed into the other, by applying an invertible linear transformation on the variables.

Quadratic forms over finite fields have been classified in the sense that each quadratic form is equivalent to a canonical one. Following \cite[Table 1]{McE}, we describe the canonical quadratic forms over finite fields with even and odd characteristic in the next two propositions. Note that we use $\Tr$ to denote the absolute trace function defined on a finite field.

\begin{proposition}\label{prop-canonicaleven}
Let $q$ be an even prime power. For $1 \le r \le m$, each quadratic form from $\Fq^m$ to $\Fq$ with rank $r$ is equivalent to one of the following canonical quadratic forms.

If $r$ is odd,
\begin{itemize}
\item[(1)] $\sum_{i=1}^{(r-1)/2} x_{2i-1}x_{2i}+x_r^2$.
\end{itemize}

If $r$ is even,
\begin{itemize}
\item[(2)] $\sum_{i=1}^{r/2} x_{2i-1}x_{2i}$.
\item[(3)] $\sum_{i=1}^{r/2-1} x_{2i-1}x_{2i}+x_{r-1}^2+x_{r-1}x_r+\la x_r^2$, where $\la$ is a nonzero element of $\Fq$ and satisfies $\Tr(\la)=1$.
\end{itemize}
\end{proposition}

\begin{proposition}\label{prop-canonicalodd}
Let $q$ be an odd prime power. For $1 \le r \le m$, each quadratic form from $\Fq^m$ to $\Fq$ with rank $r$ is equivalent to one of the following canonical quadratic forms.

If $r$ is odd,
\begin{itemize}
\item[(1)] $\sum_{i=1}^{(r-1)/2} x_{2i-1}x_{2i}+x_r^2$.
\item[(2)] $\sum_{i=1}^{(r-1)/2} x_{2i-1}x_{2i}+\la x_r^2$, $\la$ is a nonsquare of $\Fq$.
\end{itemize}

If $r$ is even,
\begin{itemize}
\item[(3)] $\sum_{i=1}^{r/2-1} x_{2i-1}x_{2i}+x_{r-1}^2-x_r^2$.
\item[(4)] $\sum_{i=1}^{r/2-1} x_{2i-1}x_{2i}+x_{r-1}^2-\la x_r^2$, $\la$ is a nonsquare of $\Fq$.
\end{itemize}
\end{proposition}

Now, we are ready to define the type of a quadratic form.

\begin{definition}\label{def-canonquad}
Let $Q$ be a quadratic form over $\Fq$.
\begin{itemize}
\item[(1)] When $q$ is even and $Q$ has even rank, $Q$ is of \emph{type 1} if it is equivalent to the canonical quadratic form in Proposition~\ref{prop-canonicaleven}(2) and of \emph{type -1} if it is equivalent to the canonical quadratic form in Proposition~\ref{prop-canonicaleven}(3).
\item[(2)] When $q$ is odd, $Q$ is of \emph{type 1} if it is equivalent to the canonical quadratic form in Proposition~\ref{prop-canonicalodd}(1) or \ref{prop-canonicalodd}(3), and of \emph{type -1} if it is equivalent to the canonical quadratic form in Proposition~\ref{prop-canonicalodd}(2) or \ref{prop-canonicalodd}(4).
\end{itemize}
A zero quadratic form, whose rank is $0$, is defined to be of type $1$.
\end{definition}

Combining Propositions~\ref{prop-canonicaleven} and~\ref{prop-canonicalodd}, we can see that up to equivalence, a quadratic form $Q$ over finite field $\Fq$ is determined by its rank and type, except when $q$ is even and $Q$ has odd rank, in which the rank solely determines the quadratic form.

Next, we introduce more notations. For $2 \le 2j \le m$ and $\tau \in \{1,-1\}$, we use $v_{2j,\tau}:=v_{2j,\tau}(q,m)$ to denote the number of quadratic forms from $\Fq^m$ to $\Fq$, with rank $2j$ and type $\tau$. We use $v_{0,1}:=v_{0,1}(q,m)$ to denote the number of quadratic forms from $\Fq^m$ to $\Fq$, with rank $0$ and type $1$. Moreover, for $1 \le 2j+1 \le m$, we use $v_{2j+1}:=v_{2j+1}(q,m)$ to denote the number of quadratic forms from $\Fq^m$ to $\Fq$ with rank $2j+1$. Considering the action of orthogonal groups on quadratic forms over finite fields, the numbers $v_{2j,\tau}$ and $v_{2j+1}$ have been obtained in \cite{McE}.


\begin{proposition}[{\cite[Table 3]{McE}}]\label{prop-quanumber}
For quadratic forms from $\Fq^m$ to $\Fq$, we have the following.
\begin{itemize}
\item[(1)] $v_1=q^m-1$ and
 $$
   \quad v_{2j+1}=q^{j^2+j}\frac{\prod_{i=m-2j}^m (q^i-1)}{\prod_{i=1}^j(q^{2i}-1)}, \quad 3 \le 2j+1 \le m.
 $$
\item[(2)] $v_{0,1}=1$. For $2 \le 2j \le m$ and $\tau \in \{1,-1\}$,
 $$
  v_{2j,\tau}=\frac{q^{j^2}(q^{j}+\tau)}{2}\frac{\prod_{i=m-2j+1}^m (q^i-1)}{\prod_{i=1}^j(q^{2i}-1)}.
 $$

\end{itemize}
\end{proposition}

\begin{remark}
Let $q$ be an odd prime power. For $1 \le 2j+1 \le m$ and $\tau \in \{1,-1\}$, the number of quadratic forms from $\Fq^m$ to $\Fq$ with rank $2j+1$ and type $\tau$ was also obtained in \cite[Table 3]{McE}. We only list $v_{2j+1}$ in Proposition~\ref{prop-quanumber}, which is sufficient for our computation below.
\end{remark}

Let $f$ be a function from $\Fq^m$ to $\Fq$. Define
$$
N(f)=|\{{\bf x} \in \Fq^m \mid f({\bf x})=0 \}|.
$$
Let $\cL$ be the set of all linear functions from $\Fq^m$ to $\Fq$. For each $c \in \Fq$, define $\ol{c}$ to be a constant function from $\Fq^m$ to $\Fq$, which sends each element of $\Fq^m$ to $c$. Let $Q$ be a quadratic form from $\Fq^m$ to $\Fq$. For each $c \in \Fq$, the multiset
$$
\{ N(f) \mid f \in \{Q+L+\ol{c} \mid L \in \cL\} \}
$$
has been computed in \cite[Appendix A]{Li}. Consequently, we can determine the multiset
$$
\{ N(f) \mid f \in \{Q+L+\ol{c} \mid L \in \cL, c \in \Fq \} \},
$$
which will be used in the computation of weight distributions. It is worthy noting that when $q$ is odd, we used different canonical quadratic forms in \cite{Li}. Thus, in order to exploit Proposition~\ref{prop-quanumber} and the results in \cite[Appendix A]{Li}, we have to establish the correspondence between the canonical quadratic forms in \cite[Proposition 3.8]{Li} and those in Propositions~\ref{prop-canonicaleven} and \ref{prop-canonicalodd}. Indeed, the relation is summarized in Table~\ref{tab-corres}. The proof of this relation is technical and included in the Appendix.

\begin{table}[h!]
\ra{1.3}
\begin{center}
\caption{The correspondence of different terminologies of canonical quadratic forms in Definition~\ref{def-canonquad} and \cite[Proposition 3.8]{Li}}
\label{tab-corres}
\vspace{4pt}
\begin{tabular}{|c|c|c|}
\hline
\multirow{3}{*}{} & terminologies of & terminologies of \\
        &canonical quadratic forms  & canonical quadratic forms  \\
        &in Definition~\ref{def-canonquad} & in \cite[Proposition 3.8]{Li} \\ \hline
\multirow{3}{*}{$q$ even}& odd rank $r$  &  odd rank $r$, type $1$  \\
                         & even rank $r$, type $1$  &  even rank $r$, type $0$  \\
                         & even rank $r$, type $-1$  &  even rank $r$, type $2$  \\ \hline
\multirow{2}{*}{$q$ odd}& odd rank $r$, type $\tau$, $\tau=\pm1$ & odd rank $r$, type $\tau\eta(-1)^{\frac{r-1}{2}}$   \\
                         & even rank $r$, type $\tau$, $\tau=\pm1$ & even rank $r$, type $\tau\eta(-1)^{\frac{r}{2}}$   \\ \hline
\end{tabular}
\end{center}
\end{table}

When $q$ is odd, we use $S$ to denote the set of nonzero squares in $\Fq$ and $NS$ the set of nonsquares in $\Fq$. Following Table~\ref{tab-corres}, we can rephrase \cite[Lemmas A.2, A.4]{Li} in the next two propositions.

\begin{proposition}\label{prop-diseven}
Let $q$ be an even prime power and $Q$ a quadratic form from $\Fq^m$ to $\Fq$. Let $\cL$ be the set of all linear functions from $\Fq^m$ to $\Fq$ and $c \in \Fq$. Suppose $f$ ranges over $\{Q+L+\ol{c} \mid L \in \cL\}$. Then the following holds.

\noindent

(1) Let $Q$ have odd rank $r$. If $c=0$, then
$$
N(f)=\left\{
\begin{aligned}
  &q^{m-1} & \mbox{$q^m-q^{r}+q^{r-1}$ times,} \\
  &q^{m-1}\pm q^{m-\frac{r+1}{2}} & \mbox{$(q-1)\frac{q^{r-1}\pm q^{\frac{r-1}{2}}}{2}$ times.}
\end{aligned}\right.
$$
If $c \ne 0$, then
$$
N(f)=\left\{
\begin{aligned}
  &q^{m-1} & \mbox{$q^m-q^{r}+q^{r-1}$ times,} \\
  &q^{m-1}\pm q^{m-\frac{r+1}{2}} & \mbox{$\frac{q^{r}-q^{r-1}\mp q^{\frac{r-1}{2}}}{2}$ times.}
\end{aligned}\right.
$$

(2) Let $Q$ have even rank $r \ge 2$ and type $\tau$. If $c=0$, then
$$
N(f)=\left\{
\begin{aligned}
  &q^{m-1} & \mbox{$q^m-q^{r}$ times,} \\
  &q^{m-1}+\tau q^{m-\frac{r+2}{2}}(q-1) & \mbox{$q^{r-1}+\tau q^{\frac{r-2}{2}}(q-1)$ times,}\\
  &q^{m-1}-\tau q^{m-\frac{r+2}{2}} & \mbox{$(q-1)(q^{r-1}-\tau q^{\frac{r-2}{2}})$ times.}
\end{aligned}\right.
$$
If $c \ne 0$, then
$$
N(f)=\left\{
\begin{aligned}
  &q^{m-1} & \mbox{$q^m-q^{r}$ times,} \\
  &q^{m-1}+\tau q^{m-\frac{r+2}{2}}(q-1) & \mbox{$q^{r-1}-\tau q^{\frac{r-2}{2}}$ times,}\\
  &q^{m-1}-\tau q^{m-\frac{r+2}{2}} & \mbox{$(q-1)q^{r-1}+\tau q^{\frac{r-2}{2}}$ times.}
\end{aligned}\right.
$$
\end{proposition}

\begin{proposition}\label{prop-disodd}
Let $q$ be an odd prime power and $Q$ a quadratic form from $\Fq^m$ to $\Fq$. Let $\cL$ be the set of all linear functions from $\Fq^m$ to $\Fq$ and $c \in \Fq$. Suppose $f$ ranges over $\{Q+L+\ol{c} \mid L \in \cL\}$. Then the following holds.

\noindent

(1) Let $Q$ have odd rank $r$ and type $\tau$. If $c=0$, then
$$
N(f)=\left\{
\begin{aligned}
  &q^{m-1} & \mbox{$q^m-q^{r}+q^{r-1}$ times,} \\
  &q^{m-1}\pm\tau q^{m-\frac{r+1}{2}} & \mbox{$\frac{(q-1)}{2}(q^{r-1}\pm\tau q^{\frac{r-1}{2}})$ times.}
\end{aligned}\right.
$$
If $c \in S$, then
$$
N(f)=\left\{
\begin{aligned}
  &q^{m-1} & \mbox{$q^m-q^{r}+q^{r-1}+\tau q^{\frac{r-1}{2}}$ times,} \\
  &q^{m-1}+\tau q^{m-\frac{r+1}{2}} & \mbox{$\frac{q-1}{2}q^{r-1}-\tau q^{\frac{r-1}{2}}$ times,} \\
  &q^{m-1}-\tau q^{m-\frac{r+1}{2}} & \mbox{$\frac{q-1}{2}q^{r-1}$ times.}
\end{aligned}\right.
$$
If $c \in NS$, then
$$
N(f)=\left\{
\begin{aligned}
  &q^{m-1} & \mbox{$q^m-q^{r}+q^{r-1}-\tau q^{\frac{r-1}{2}}$ times,} \\
  &q^{m-1}+\tau q^{m-\frac{r+1}{2}} & \mbox{$\frac{q-1}{2}q^{r-1}$ times,} \\
  &q^{m-1}-\tau q^{m-\frac{r+1}{2}} & \mbox{$\frac{q-1}{2}q^{r-1}+\tau q^{\frac{r-1}{2}}$ times.}
\end{aligned}\right.
$$

(2) Let $Q$ have even rank $r \ge 2$ and type $\tau$. If $c=0$, then
$$
N(f)=\left\{
\begin{aligned}
  &q^{m-1} & \mbox{$q^m-q^{r}$ times,} \\
  &q^{m-1}+\tau q^{m-\frac{r+2}{2}}(q-1) &\mbox{$q^{r-1}+\tau q^{\frac{r-2}{2}}(q-1)$ times,}\\
  &q^{m-1}-\tau q^{m-\frac{r+2}{2}} & \mbox{$(q-1)(q^{r-1}-\tau q^{\frac{r-2}{2}})$ times.}
\end{aligned}\right.
$$
If $c\ne0$, then
$$
N(f)=\left\{
\begin{aligned}
  &q^{m-1} & \mbox{$q^m-q^{r}$ times,} \\
  &q^{m-1}+\tau q^{m-\frac{r+2}{2}}(q-1) &\mbox{$q^{r-1}-\tau q^{\frac{r-2}{2}}$ times,}\\
  &q^{m-1}-\tau q^{m-\frac{r+2}{2}} & \mbox{$(q-1)q^{r-1}+\tau q^{\frac{r-2}{2}}$ times.}
\end{aligned}\right.
$$
\end{proposition}

Let $Q$ be a quadratic form from $\Fq^m$ to $\Fq$ and $c \in \Fq$, Propositions~\ref{prop-diseven} and \ref{prop-disodd} describe the multiset
$$
\{ N(f) \mid f \in \{Q+L+\ol{c} \mid L \in \cL\} \}.
$$
Now we are ready to determine the multiset $\{ N(f) \mid f \in \{Q+L+\ol{c} \mid L \in \cL, c \in \Fq \} \}$. The following result follows from Propositions~\ref{prop-diseven} and \ref{prop-disodd}.

\begin{proposition}\label{prop-cosetzero}
Let $Q$ be a quadratic form from $\Fq^m$ to $\Fq$. Suppose $f$ ranges over $\{ Q+L+\ol{c} \mid L \in \cL, c \in \Fq \}$. Then the following holds.
\begin{itemize}
\item[(1)] If $Q$ has odd rank $r$, then
$$
N(f)=\left\{
\begin{aligned}
  &q^{m-1} & \mbox{$q^{m+1}-q^{r+1}+q^r$ times,} \\
  &q^{m-1}\pm q^{m-\frac{r+1}{2}} & \mbox{$\frac{q-1}{2}q^r$ times.}
\end{aligned}\right.
$$
\item[(2)] If $Q$ has even rank $r$ and type $\tau$, then
$$
N(f)=\left\{
\begin{aligned}
  &q^{m-1} & \mbox{$q^{m+1}-q^{r+1}$ times,} \\
  &q^{m-1}+\tau q^{m-\frac{r+2}{2}}(q-1) & \mbox{$q^r$ times,} \\
  &q^{m-1}-\tau q^{m-\frac{r+2}{2}} & \mbox{$(q-1)q^r$ times.}
\end{aligned}\right.
$$
\end{itemize}
\end{proposition}

Note that when $Q$ has odd rank, the multiset $\{ N(f) \mid f \in \{Q+L+\ol{c} \mid L \in \cL, c \in \Fq \} \}$ does not depend on the type of $Q$.

We finally mention the following well known result concerning $N(Q)$, where $Q$ is a quadratic form from $\Fq^m$ to $\Fq$.

\begin{proposition}\label{prop-quadzero}
Let $Q$ be a quadratic form from $\Fq^m$ to $\Fq$. We have
$$
N(Q)=
\begin{cases}
  q^{m-1} & \mbox{if $Q$ has odd rank,} \\
  q^{m-1}+\tau q^{m-\frac{r+2}{2}}(q-1) & \mbox{if $Q$ has even rank $r$ and type $\tau$.}
\end{cases}
$$
\end{proposition}
\begin{proof}
The conclusion is a direct consequence of \cite[Theorems 6.26, 6.27, 6.32]{LN} and Table~\ref{tab-corres}.
\end{proof}

\section{The weight distribution of second order Reed-Muller codes and their relatives}\label{sec3}

In this section, we compute the weight distributions of second order Reed-Muller codes and their relatives. For this purpose, a brief introduction to Reed-Muller codes, homogeneous Reed-Muller codes and projective Reed-Muller codes is also included. For a more detailed treatment, see \cite[Section 5]{AK}, \cite[Chapters 13,14,15]{MS} for Reed-Muller codes, \cite{MDCE} for homogeneous Reed-Muller codes and \cite{La1,La2,Soren} for projective Reed-Muller codes. For the basic knowledge of coding theory, please refer to \cite{MS}. We only mention a few notations below.

Let $\cC$ be a code of length $n$. For ${\bf c} \in \cC$, $\wt({\bf c})$ is the Hamming weight of the codeword ${\bf c}$. For $0 \le i \le n$, we use $A_i:=A_i(\cC)$ to denote the number of codewords in $\cC$ with Hamming weight $i$. The sequence $(A_0,A_1,\ldots,A_n)$ is called the \emph{weight distribution} of $\cC$. The \emph{weight enumerator} of $\cC$ is a polynomial $\sum_{i=0}^n A_iZ^i$, which gives a compact expression of the weight distribution.

\subsection{Reed-Muller codes}

Let $q$ be a prime power. Let $\Fq^m=\{{\bf v}_i \mid 0 \le i \le q^m-1\}$ be an $m$-dimensional vector space over $\Fq$. We use $\cP_q(r,m)$ to denote the vector space of all polynomials over $\Fq$ with $m$ variables and degree at most $r$. Let $0 \le r \le m(q-1)$, the $r$-th order $q$-ary Reed-Muller code of length $q^m$ is defined as
$$
\{ (f({\bf v}_i))_{i=0}^{q^m-1} \mid f \in \cP_q(r,m) \},
$$
which is denoted by $\RM_q(r,m)$. In particular, when $r=2$, the second order $q$-ary Reed-Muller code $\RM_q(2,m)$ has the following parameters \cite[Theorem 5.4.1, Corollary 5.5.4]{AK}:
\begin{equation}\label{eqn-RMpara}
[n,k,d]=\begin{cases}
  [2^m,\frac{m^2+m+2}{2},2^{m-2}] & \mbox{if $q=2$, $m \ge 2$,} \\
  [q^m,\frac{m^2+3m+2}{2},(q-2)q^{m-1}] & \mbox{if $q>2$.}
\end{cases}
\end{equation}

We use ${\bf 0}$ and ${\bf 1}$ to denote the all-zero and all-one vector of length $q^m$ over $\Fq$. Let $\cQ$ be the set of all quadratic forms from $\Fq^m$ to $\Fq$. By definition, the second order Reed-Muller code $\RM_q(2,m)$ can be decomposed into a disjoint union of cosets:
\begin{equation}\label{eqn-RMdecom}
\RM_q(2,m)=\begin{cases}
\bigcup_{Q \in \cQ} \big((Q({\bf v}_i))_{i=0}^{q^m-1}+\{{\bf 0}, {\bf 1}\}\big) & \mbox{if $q=2$, $m \ge 2$} \\
\bigcup_{Q \in \cQ} \big((Q({\bf v}_i))_{i=0}^{q^m-1}+\RM_q(1,m)\big) & \mbox{if $q>2$.}
\end{cases}
\end{equation}
Note that when $q=2$, we have $x=x^2$ for each $x \in \Ft$. Thus, each linear function from $\Ft^m$ to $\Ft$ is actually a quadratic form. This explains the distinct decompositions in \eqref{eqn-RMdecom} for $q=2$ and $q>2$, as well as the different dimensions in \eqref{eqn-RMpara}.

Thus, to compute the weight distribution of $\RM_q(2,m)$, it suffices to compute the weight distribution of $(Q({\bf v}_i))_{i=0}^{q^m-1}+\{{\bf 0}, {\bf 1}\}$ or $(Q({\bf v}_i))_{i=0}^{q^m-1}+\RM_q(1,m)$, for each $Q \in \cQ$. The following proposition says the weight distribution of $(Q({\bf v}_i))_{i=0}^{q^m-1}+\RM_q(1,m)$ depends only on the rank and type of $Q$.

\begin{proposition}\label{prop-cosetweight}
Let $Q$ be a quadratic form from $\Fq^m$ to $\Fq$. Suppose ${\bf c}$ ranges over $(Q({\bf v}_i))_{i=0}^{q^m-1}+\RM_q(1,m)$. Then the following holds.
\begin{itemize}
\item[(1)] If $Q$ has odd rank $r$, then
$$
\wt({\bf c})=\left\{
\begin{aligned}
  &q^m-q^{m-1} & \mbox{$q^{m+1}-q^{r+1}+q^r$ times,} \\
  &q^m-q^{m-1}\pm q^{m-\frac{r+1}{2}} & \mbox{$\frac{q-1}{2}q^r$ times.}
\end{aligned}\right.
$$
\item[(2)] If $Q$ has even rank $r$ and type $\tau$, then
$$
\wt({\bf c})=\left\{
\begin{aligned}
  &q^m-q^{m-1} & \mbox{$q^{m+1}-q^{r+1}$ times,} \\
  &q^m-q^{m-1}-\tau q^{m-\frac{r+2}{2}}(q-1) & \mbox{$q^r$ times,} \\
  &q^m-q^{m-1}+\tau q^{m-\frac{r+2}{2}} & \mbox{$(q-1)q^r$ times.}
\end{aligned}\right.
$$
\end{itemize}
\end{proposition}
\begin{proof}
Recall that $\cL$ is the set of all linear functions from $\Fq^m$ to $\Fq$. By definition, the first order Reed-Muller code
$$
\RM_q(1,m)=\{ (g({\bf v}_i))_{i=0}^{q^m-1} \mid g \in \{ L+\ol{c} \mid L \in \cL, c \in \Fq \} \}.
$$
Thus, the weight distribution of $(Q({\bf v}_i))_{i=0}^{q^m-1}+\RM_q(1,m)$ is the multiset
$$
\{q^m-N(f) \mid f \in \{Q+L+\ol{c} \mid L \in \cL, c \in \Fq\} \}.
$$
Therefore, the conclusion follows from Proposition~\ref{prop-cosetzero}.
\end{proof}

Now we are ready to compute the weight distribution of second order Reed-Muller codes.

\begin{theorem}\label{thm-RM}
The weight distribution of the second order Reed-Muller code $\RM_q(2,m)$ is listed in Table~\ref{tab-RM1} for $q=2$ and in Table~\ref{tab-RM2} for $q>2$, where $v_{2j,1}$, $v_{2j,-1}$, $v_{2j+1}$ are defined in Proposition~\ref{prop-quanumber} and $v_{2j+1}=0$ if $2j+1>m$.
\end{theorem}
\begin{proof}
When $q=2$, the conclusion follows from \eqref{eqn-RMdecom} and Propositions~\ref{prop-quanumber} and \ref{prop-quadzero}. When $q>2$, the conclusion follows from \eqref{eqn-RMdecom} and Propositions~\ref{prop-quanumber} and \ref{prop-cosetweight}.
\end{proof}

\begin{table}[h!]
\ra{2.25}
\begin{center}
\caption{Weight distribution of $\RM_2(2,m)$, $m \ge 2$}
\label{tab-RM1}
\vspace{4pt}
\begin{tabular}{|c|c|}
\hline
Weight & Frequency \\ \hline
$0$  &  $1$  \\ \hline
$2^{m-1}$  &  $2\big(2^m-1+\sum_{j=1}^{\lf (m-1)/2 \rf} 2^{j^2+j}\frac{\prod_{i=m-2j}^m (2^i-1)}{\prod_{i=1}^j(2^{2i}-1)}\big)$  \\ \hline
$2^{m-1}\pm 2^{m-j-1}, 1 \le j \le \lf \frac{m}{2} \rf$  &  $2^{j^2+j}\frac{\prod_{i=m-2j+1}^m (2^i-1)}{\prod_{i=1}^j (2^{2i}-1)}$ \\ \hline
$2^m$  &  $1$  \\ \hline
\end{tabular}
\end{center}
\end{table}

\begin{table}[h!]
\ra{1.5}
\begin{center}
\caption{Weight distribution of $\RM_q(2,m)$, $q>2$}
\label{tab-RM2}
\vspace{4pt}
\begin{tabular}{|c|c|}
\hline
Weight & Frequency \\ \hline
$0$  &  $1$  \\ \hline
$q^m-q^{m-1}-\tau q^{m-j-1}(q-1)$,   &  \multirow{2}{*}{$\frac{q^{j^2+2j}(q^{j}+\tau)}{2}\frac{\prod_{i=m-2j+1}^m (q^i-1)}{\prod_{i=1}^j(q^{2i}-1)}$}  \\
$1 \le j \le \lf \frac{m}{2} \rf$, $\tau \in \{1,-1\}$ & \\ \hline
$q^m-q^{m-1}+\tau q^{m-j-1}$,   &  \multirow{2}{*}{$(q-1)q^{2j} v_{2j,\tau}+\frac{q-1}{2}q^{2j+1}v_{2j+1}$}  \\
$1 \le j \le \lf \frac{m}{2} \rf$, $\tau \in \{1,-1\}$ & \\ \hline
$q^{m}- 2q^{m-1}$ &  $\frac{q(q-1)(q^m-1)}{2}$  \\ \hline
\multirow{2}{*}{$q^m-q^{m-1}$} & $q^{\frac{m^2+3m+2}{2}}-\sum_{j=1}^{\lf m/2 \rf} q^{2j+1} (v_{2j,1}+v_{2j,-1})$ \\
& $-\sum_{j=0}^{\lf m/2 \rf}(q-1)q^{2j+1} v_{2j+1}-q$ \\ \hline
$q^m$  &  $\frac{q(q-1)(q^m-1)}{2}+q-1$  \\ \hline
\end{tabular}
\end{center}
\end{table}

\begin{remark}
Table~\ref{tab-RM1} reproduces the results in \cite{SB} (see also \cite[Chapter 15, Theorem 8]{MS}). Still, we do not have a compact formula on the frequency of codewords with weight $2^{m-1}$. Similarly, in Table~\ref{tab-RM2}, no compact formula on the frequency of codewords with weight $q^m-q^{m-1}$ is available.
\end{remark}

\begin{example}
Numerical experiment shows that the second order binary Reed-Muller code $\RM_2(2,7)$ has weight enumerator
\begin{align*}
1&+10668Z^{32}+5291328Z^{48}+112881664Z^{56}+300503590Z^{64}+\\
 &112881664Z^{72}+5291328Z^{80}+10668Z^{96}+Z^{128},
\end{align*}
which is consistent with Table~\ref{tab-RM1}.
\end{example}

\begin{example}
Numerical experiment shows that the second order ternary Reed-Muller code $\RM_3(2,4)$ has weight enumerator
\begin{align*}
1&+240Z^{27}+14040Z^{36}+519480Z^{45}+1705860Z^{48}+2729376Z^{51}+4062720Z^{54}+\\
 &3411720Z^{57}+1364688Z^{60}+533520Z^{63}+7020Z^{72}+242Z^{81},
\end{align*}
which is consistent with Table~\ref{tab-RM2}.
\end{example}

\subsection{Homogeneous Reed-Muller codes}

Some variations of Reed-Muller codes were discussed in the literature. As an attempt to find subcodes of Reed-Muller codes with large minimum distances, the concept of homogeneous Reed-Muller codes was proposed \cite{McE,MDCE}. Let $\Fq^m=\{{\bf v}_i \mid 0 \le i \le q^m-1\}$ be an $m$-dimensional vector space over $\Fq$. We use $\cP_q^h(r,m)$ to denote the vector space of all homogeneous polynomials over $\Fq$ with $m$ variables and degree $r$. Let $0 \le r \le m(q-1)$, the $r$-th order $q$-ary homogeneous Reed-Muller code of length $q^m$ is defined as
$$
\{ (f({\bf v}_i))_{i=0}^{q^m-1} \mid f \in \cP_q^h(r,m) \},
$$
which is denoted by $\HRM_q(r,m)$. Without loss of generality, suppose that ${\bf v}_0$ is the zero vector in $\Fq^m$. Then by definition, the first coordinate of $\HRM_q(r,m)$ is always $0$. Thus, in some literature, the homogeneous Reed-Muller code is defined to be the punctured code of $\HRM_q(r,m)$ with the first coordinate deleted (see \cite{Ber,McE} for instance). When $r=2$, the second order $q$-ary homogeneous Reed-Muller code $\HRM_q(2,m)$ has the following parameters \cite[Proposition 4]{Ber}:
$$
[n,k,d]=\begin{cases}
          [q,1,q-1] & \mbox{if $q>2$, $m=1$,} \\
          [q^m,\frac{m(m+1)}{2},(q-1)^2q^{m-2}] & \mbox{otherwise.}
        \end{cases}
$$
Note that
$$
\HRM_q(2,m)=\{ (Q({\bf v}_i))_{i=0}^{q^m-1} \mid Q \in \cQ \},
$$
where $\cQ$ is the set of all quadratic forms from $\Fq^m$ to $\Fq$. The weight distribution of $\HRM_q(2,m)$ follows easily from Propositions~\ref{prop-quanumber} and \ref{prop-quadzero}.

\begin{theorem}\label{thm-HRM}
The weight distribution of the second order homogeneous Reed-Muller code $\HRM_q(2,m)$ is listed in Table~\ref{tab-homoRM}.
\end{theorem}

\begin{remark}
The weight distribution of $\HRM_q(2,m)$ has essentially been obtained in \cite[Table 6]{McE}, in which the punctured code of $\HRM_q(2,m)$ was considered. Thus, the above theorem is not new. We list the weight distribution in Table~\ref{tab-homoRM}, where a few typos in \cite[Table 6]{McE} are corrected. 
\end{remark}

\begin{table}
\ra{2.2}
\begin{center}
\caption{Weight distribution of $\HRM_q(2,m)$}
\label{tab-homoRM}
\vspace{4pt}
\begin{tabular}{|c|c|}
\hline
Weight & Frequency \\ \hline
$0$  &  $1$  \\ \hline
$q^m-q^{m-1}$  &  $q^m-1+\sum_{j=1}^{\lf (m-1)/2 \rf} q^{j^2+j}\frac{\prod_{i=m-2j}^m (q^i-1)}{\prod_{i=1}^j(q^{2i}-1)}$  \\ \hline
$q^{m}-q^{m-1}-\tau q^{m-j-1}(q-1)$,  &  \multirow{2}{*}{$\frac{q^{j^2}(q^{j}+\tau)}{2}\frac{\prod_{i=m-2j+1}^m (q^i-1)}{\prod_{i=1}^j(q^{2i}-1)}$}  \\
$1 \le j \le \lf \frac{m}{2} \rf$, $\tau \in \{1,-1\}$ & \\ \hline
\end{tabular}
\end{center}
\end{table}

\begin{example}
Numerical experiment shows that the second order ternary homogeneous Reed-Muller code $\HRM_3(2,4)$ has weight enumerator
$$
1+1560Z^{36}+21060Z^{48}+18800Z^{54}+16848Z^{60}+780Z^{72},
$$
which is consistent with Table~\ref{tab-homoRM}.
\end{example}

\subsection{Projective Reed-Muller codes}

Another variation of Reed-Muller codes adopts a geometric viewpoint, in which the Reed-Muller codes are regarded as codes defined on an affine space. As a natural projective analogue of Reed-Muller codes, the concept of projective Reed-Muller codes was proposed by Lachaud \cite{La1,La2}. Consider an $(m+1)$-dimensional vector space $\Fq^{m+1}=\{{\bf w}_i \mid 0 \le i \le q^{m+1}-1\}$ over $\Fq$, where ${\bf w_0}$ is the zero vector. We introduce an equivalence relation among nonzero elements of $\Fq^{m+1}$ as follows. For nonzero ${\bf x}=(x_0,x_1,\ldots,x_m)$ and ${\bf y}=(y_0,y_1,\ldots,y_m)$ in $\Fq^{m+1}$, we define ${\bf x}\sim{\bf y}$ if and only if there exists a nonzero $\la \in \Fq$, such that $(x_0,x_1,\ldots,x_m)=(\la y_0, \la y_1,\ldots,\la y_m)$. It is easy to see that this relation is indeed an equivalence relation partitioning all nonzero elements of $\Fq^{m+1}$. Without loss of generality, for each $1 \le i,j \le \frac{q^{m+1}-1}{q-1}$, $i \ne j$, assume that ${\bf w}_i$ and ${\bf w}_j$ belong to distinct equivalence classes. Thus, $\{{\bf w}_i \mid 1 \le i \le \frac{q^{m+1}-1}{q-1} \}$ is a set of representatives of the equivalence classes in $\Fq^{m+1}$. Let $0 \le r \le (m+1)(q-1)$, the $r$-th order $q$-ary projective Reed-Muller code of length $\frac{q^{m+1}-1}{q-1}$ is defined as
$$
\{ (f({\bf w}_i))_{i=1}^{(q^{m+1}-1)/(q-1)} \mid f \in \cP_q^h(r,m+1) \},
$$
which is denoted by $\PRM_q(r,m)$. When $q=2$, we know that $\PRM_2(r,m)$ is identical to the punctured code of $\HRM_2(r,m+1)$, in which the coordinate associated with the zero vector is deleted. When $r=2$, the second order $q$-ary projective Reed-Muller code $\PRM_q(2,m)$ has the following parameters \cite[Theorem 1]{Soren}:
$$
[n,k,d]=[\frac{q^{m+1}-1}{q-1},\frac{(m+1)(m+2)}{2},(q-1)q^{m-1}].
$$

Note that there is a one-to-one correspondence between codewords of $\PRM_q(2,m)$ and codewords of $\HRM_q(2,m+1)$:
\begin{align*}
\PRM_q(2,m) & \longleftrightarrow \HRM_q(2,m+1) \\
(Q({\bf w}_i))_{i=1}^{(q^{m+1}-1)/(q-1)} &\longleftrightarrow (Q({\bf w}_i))_{i=0}^{q^{m+1}-1},
\end{align*}
where $Q \in \cQ$ and $\cQ$ is the set of all quadratic forms from $\Fq^{m+1}$ to $\Fq$. By the property of quadratic forms, we have $\wt((Q({\bf w}_i))_{i=1}^{(q^{m+1}-1)/(q-1)})=\wt((Q({\bf w}_i))_{i=0}^{q^{m+1}-1})/(q-1)$ for each $Q \in \cQ$. Consequently, the weight distribution of $\PRM_q(2,m)$ follows easily from Theorem~\ref{thm-HRM}.

\begin{theorem}\label{thm-PRM}
The weight distribution of the second order projective Reed-Muller code $\PRM_q(2,m)$ is listed in Table~\ref{tab-proRM}.
\end{theorem}

\begin{table}
\ra{2.2}
\begin{center}
\caption{Weight distribution of $\PRM_q(2,m)$}
\label{tab-proRM}
\vspace{4pt}
\begin{tabular}{|c|c|}
\hline
Weight & Frequency \\ \hline
$0$  &  $1$  \\ \hline
$q^{m}$  &  $q^{m+1}-1+\sum_{j=1}^{\lf m/2 \rf} q^{j^2+j}\frac{\prod_{i=m-2j+1}^{m+1} (q^i-1)}{\prod_{i=1}^j(q^{2i}-1)}$  \\ \hline
$q^{m}-\tau q^{m-j}$,  &  \multirow{2}{*}{$\frac{q^{j^2}(q^{j}+\tau)}{2}\frac{\prod_{i=m-2j+2}^{m+1} (q^i-1)}{\prod_{i=1}^j(q^{2i}-1)}$}  \\
$1 \le j \le \lf \frac{m+1}{2} \rf$, $\tau \in \{1,-1\}$ & \\ \hline
\end{tabular}
\end{center}
\end{table}

\begin{example}
Numerical experiment shows that the second order ternary projective Reed-Muller code $\PRM_3(2,4)$ has weight enumerator
$$
1+14520Z^{54}+2548260Z^{72}+9740258Z^{81}+2038608Z^{90}+7260Z^{108},
$$
which is consistent with Table~\ref{tab-proRM}.
\end{example}

\section{Conclusion}\label{sec4}

In this paper, we compute the weight distributions of second order $q$-ary Reed-Muller code, second order $q$-ary homogeneous Reed-Muller codes and second order $q$-ary projective Reed-Muller codes. Besides, in Propositions~\ref{prop-diseven} and \ref{prop-disodd}, for each $c \in \Fq$, the multiset
$$
\{ N(f) \mid f \in \{Q+L+\ol{c} \mid L \in \cL\} \}
$$
has been determined, where $Q$ is a quadratic form from $\Fq^m$ to $\Fq$ and $\cL$ is the set of all linear functions from $\Fq^m$ to $\Fq$. In this sense, the results in Propositions~\ref{prop-diseven} and \ref{prop-disodd} are of independent interest.

\section*{Acknowledgement}

Shuxing Li is supported by the Alexander von Humboldt Foundation. He is indebted to Professor Cunsheng Ding for pointing out the mistakes in McEliece's paper and many helpful suggestions. He wishes to thank Professor Alexander Pott for his careful reading and very helpful comments.

\section*{Appendix}

In this appendix, we show that the correspondence in Table~\ref{tab-corres} holds true. Note that when $q$ is even, the terminologies in Definition~\ref{def-canonquad} and \cite[Proposition 3.8]{Li} simply have different names and are essentially the same. Hence, it suffices to consider the $q$ odd case. The following is a preparatory lemma.

\begin{lemma}\label{lem-prep}
Let $q$ be an odd prime power and $\la$ a nonsquare of $\Fq$. Then the quadratic form $x_1^2+x_2^2$ is equivalent to $\la(y_1^2+y_2^2)$.
\end{lemma}
\begin{proof}
It is easy to see that there exist $\la_1, \la_2 \in \Fq$, such that $\la_1 \ne \la_2$ and $\la_1^2+\la_2^2=\la$. The conclusion follows by applying an invertible linear transformation satisfying $x_1=\la_1y_1+\la_2y_2$ and $x_2=-\la_2y_1+\la_1y_2$.
\end{proof}

When $q$ is odd, we use $\eta$ to denote the quadratic multiplicative character of $\Fq$. We have the following lemma.

\begin{lemma}\label{lem-equiv}
Let $q$ be an odd prime power. For ${\bf x}=(x_1,x_2,\ldots,x_m)$, let $Q({\bf x})=\sum_{i=1}^r a_ix_i^2$ be a quadratic form of rank $r$ from $\Fq^m$ to $\Fq$, where each $a_i$ is a nonzero element of $\Fq$ and $\eta(\prod_{i=1}^r a_i)=\de$. Then $Q$ is a quadratic form of rank $r$ and type $\tau$, where
$$
\tau=\begin{cases}
  -\de & \mbox{if $q \equiv 3 \bmod4$ and $r \equiv 2,3 \bmod4$,} \\
  \de & \mbox{otherwise.}
\end{cases}
$$
\end{lemma}
\begin{proof}
We only prove the case where the rank $r$ is odd, since the proof of $r$ even case is analogous. Applying Lemma~\ref{lem-prep}, we can see that $Q$ is equivalent to one of the following:
\begin{equation}\label{eqn-iniequiv}
\begin{cases}
  \sum_{i=1}^r y_i^2 & \mbox{if $\de=1$,} \\
  \sum_{i=1}^{r-1} y_i^2+\la y_r^2 & \mbox{if $\de=-1$,}
\end{cases}
\end{equation}
where $\la$ is a nonsquare of $\Fq$. Next, we divide the proof in two cases, in which $q \equiv 1 \bmod4$ and $q \equiv 3 \bmod4$, respectively.

{\bf Case I: $q \equiv 1 \bmod 4$.} In this case, $-1$ is a square of $\Fq$. When $\de=1$, by \eqref{eqn-iniequiv}, $Q$ is equivalent to
$$
\sum_{i=1}^{\frac{r-1}{2}} z_i^2-\sum_{i=\frac{r+1}{2}}^{r-1} z_i^2+z_r^2.
$$
Let $w_i=z_i+z_{i+\frac{r-1}{2}}$, $w_{i+\frac{r-1}{2}}=z_i-z_{i+\frac{r-1}{2}}$ for $1 \le i \le \frac{r-1}{2}$ and $w_r=z_r$. We have that $Q$ is equivalent to
$$
\sum_{i=1}^{\frac{r-1}{2}} w_iw_{i+\frac{r-1}{2}}+w_r^2.
$$
Thus, $Q$ has rank $r$ and type $1$. When $\de=-1$, a similar argument shows that $Q$ has rank $r$ and type $-1$.

{\bf Case II: $q \equiv 3 \bmod 4$.} In this case, $-1$ is a nonsquare of $\Fq$. We consider two subcases where $r \equiv 1 \bmod 4$ and $r \equiv 3 \bmod 4$, respectively.

When $r \equiv 1 \bmod 4$, if $\de=1$, then by \eqref{eqn-iniequiv} and Lemma~\ref{lem-prep}, $Q$ is equivalent to
$$
\sum_{i=1}^{\frac{r-1}{2}} z_i^2-\sum_{i=\frac{r+1}{2}}^{r-1} z_i^2+z_r^2.
$$
An analogous approach as in Case I shows that $Q$ has rank $r$ and type $1$. Similarly, if $\de=-1$, we can show that $Q$ has rank $r$ and type $-1$.

When $r \equiv 3 \bmod 4$, the situation is more involved. If $\de=1$, then by \eqref{eqn-iniequiv} and Lemma~\ref{lem-prep}, $Q$ is equivalent to
$$
\sum_{i=1}^{\frac{r-3}{2}} z_i^2-\sum_{i=\frac{r-1}{2}}^{r-3} z_i^2+z_{r-2}^2-z_{r-1}^2-z_r^2.
$$
Let $w_i=z_i+z_{i+\frac{r-3}{2}}$, $w_{i+\frac{r-3}{2}}=z_i-z_{i+\frac{r-3}{2}}$ for $1 \le i \le \frac{r-3}{2}$ and $w_{r-2}=z_{r-2}+z_{r-1}$, $w_{r-1}=z_{r-2}-z_{r-1}$, $w_r=z_r$. Then we can see that $Q$ is equivalent to
$$
\sum_{i=1}^{\frac{r-3}{2}} w_iw_{i+\frac{r-3}{2}}+w_{r-2}w_{r-1}+(-1)w_r^2.
$$
Thus, $Q$ has rank $r$ and type $-1$. If $\de=-1$, then by \eqref{eqn-iniequiv} and Lemma~\ref{lem-prep}, $Q$ is equivalent to
$$
\sum_{i=1}^{\frac{r-3}{2}} z_i^2-\sum_{i=\frac{r-1}{2}}^{r-3} z_i^2+z_{r-2}^2+z_{r-1}^2-z_r^2.
$$
Let $w_i=z_i+z_{i+\frac{r-3}{2}}$, $w_{i+\frac{r-3}{2}}=z_i-z_{i+\frac{r-3}{2}}$ for $1 \le i \le \frac{r-3}{2}$ and $w_{r-2}=z_{r-1}+z_{r}$, $w_{r-1}=z_{r-1}-z_{r}$, $w_r=z_{r-2}$. Then we can see that $Q$ is equivalent to
$$
\sum_{i=1}^{\frac{r-3}{2}} w_iw_{i+\frac{r-3}{2}}+w_{r-2}w_{r-1}+w_r^2.
$$
Thus, $Q$ has rank $r$ and type $1$. Consequently, we complete the proof.
\end{proof}

Employing Lemma~\ref{lem-equiv} and comparing Definition~\ref{def-canonquad} with \cite[Proposition 3.8]{Li}, we confirm the correspondence in Table~\ref{tab-corres}.

\end{document}